\documentclass[UKenglish]{eptcs}

\usepackage{needspace}
\usepackage[latin1]{inputenc}
\usepackage{pslatex}
\usepackage{amsmath}
\usepackage{amssymb}
\usepackage{amsthm}
\usepackage{amsfonts}
\usepackage{hyperref}
\usepackage{breakurl}
\usepackage{prooftree}
\usepackage{cite}

\providecommand{\publicationstatus}{ITRS 2016}

\usepackage{xcolor}

\newif\ifmc
\mctrue 

\newtheorem{fact}{Fact}[section]
\usepackage{graphicx}
\usepackage{enumitem}
\setlist[description]{font=\normalfont\itshape}

\newtheorem{lemma}[fact]{Lemma}

\newtheorem{corollary}[fact]{Corollary}

\newtheorem{example}[fact]{Example}

\newtheorem{definition}[fact]{Definition}

\newtheorem{theorem}[fact]{Theorem}

\newcommand{\tA}{\sigma}       

\newcommand{\B}{\Gamma}

\newcommand{\db}{\displaystyle}

\renewcommand{\l}{\lambda}

\newcommand{\p}{{\sf p}}

\newcommand{\red}{\longrightarrow}

\newcommand{\set}[1]{\ensuremath{\{#1\}}}

\newcommand{\labelx}[1]{\label{#1}}

\newcommand{\Gint}[3]{[ \! [#1] \! ]^{#2}_{#3}}

\newcommand{\Tenv}{{\cal W}}
\newcommand{\Lenv}{{\cal V}}
\newcommand{\Mod}{{\cal M}}

\newcommand{\rank}{\text{rank}}

\newcommand{\trp}[3]{\langle#1,#2,#3\rangle}

\newcommand{\fmod}{{\cal F}}

\newcommand{\sub}{{\sf{s}}}

\begin{document}

 \title{Retractions in Intersection Types\thanks{Partially supported by EU H2020-644235 Rephrase project, EU H2020-644298 HyVar project, ICT COST Actions IC1201 BETTY, IC1402 ARVI, IC1405 Reversible Computation, CA1523 EUTYPES, Ateneo/CSP project RunVar and STIC-AmSud project FoQCoSS.}
}

\author{
Mario Coppo$^1$
\quad
Mariangiola Dezani-Ciancaglini$^1$
\and
Alejandro D\'{\i}az-Caro$^{2,1}$
\quad
Ines Margaria$^1$
\quad
Maddalena Zacchi$^1$
\institute{${}^1$ Dipartimento di Informatica Universit\`a di Torino,
corso Svizzera 185, 10149 Torino, Italy\\
${}^2$ CONICET \& Universidad Nacional de Quilmes, Roque S\'aenz Pe\~na 352, B1876BXD Bernal, Buenos Aires, Argentina}
}

\def\titlerunning{Retractions in Intersection Types 
}
\def\authorrunning{M. Coppo,  M. Dezani-Ciancaglini, A. D\'{\i}az-Caro, I. Margaria \& M. Zacchi}


\maketitle

%

\begin{abstract}
  This paper deals with retraction - intended as isomorphic embedding - in intersection types building  left and right inverses as terms of a $\lambda$-calculus with a $\bot$ constant. The main result is a necessary and sufficient condition  two strict intersection types must satisfy in order to assure the existence of two terms showing  the first type to be a retract of  the second one. Moreover, the characterisation of retraction in the standard intersection types is discussed.
\end{abstract}

\section{Introduction}\label{i}
\emph{Isomorphism} of types has been first discussed in the seminal paper~\cite{BL85} and then studied in various type disciplines~\cite{S83,S93,BDL92,D95,D05,FDB06,DDGT10,CDMZ13,CDMZ14,CDMZ15,CDMZ16}. Two types $\sigma$ and $\tau$ in some typed
calculus are isomorphic if there are two terms
$L$ and $R$ of types $\tau\to\sigma$ and $\sigma\to\tau$, respectively, such that the composition $L\circ R$ is  equal to the identity at type $\sigma$ and  the composition $R\circ L$ is  equal to the identity at type $\tau$.

\medskip

We claim that in programming practice and theory the notion of \emph{retraction} between types plays a central role, and it is more widespread than that of type isomorphism. Type $\sigma$ is a retract of type $\tau$ in some typed
calculus if there are two terms
$L$ and $R$ of types  $\tau\to\sigma$  and  $\sigma\to\tau$, respectively, such that the composition $L\circ R$ is  equal to the identity at type $\sigma$. Clearly $L$  is \emph{right invertible} and $R$ is \emph{left invertible}. The terms $R$ and $L$ are injective and surjective, respectively, on their domains. They are also called the {\em coder} and the {\em decoder} of type $\sigma$ in type $\tau$~\cite{P01,S08}. In fact, the term $R$ encodes the values in $\sigma$ as elements of $\tau$, while the term $L$ decodes back from $\tau$ to $\sigma$, by returning the original values.

\medskip

To the best of our knowledge, type retraction as defined above has been discussed for Curry types and higher-order types~\cite{BL85,LPS92,SU99,P01,RU02,S08,S13}, but not for intersection types. Aim of this paper is to make a first step toward the filling of this gap. We consider the $\lambda \bot$-calculus as given in~\cite[Definition 14.3.1]{B84}. For this calculus the terms having left and right inverses have been characterised~\cite{MZ83}. We reformulate this characterisation in order to simplify the study of the types which can be derived for these terms. In particular we identify a class of right inverses (the  \emph{simple right inverses}) such that if a term has a right inverse it has also a right inverse belonging to this class. All results in this paper are  given  for the $\lambda\bot$-calculus but they hold, as well, for the $\lambda$-calculus. 
\medskip

We choose to investigate retraction for the essential intersection type assignment system introduced in~\cite{B11}. This system has the same typeability power of the standard system~\cite{barecoppdeza83} and a less permissive type syntax. This restriction better fits the technical development of the present paper, as discussed in the Conclusion.
Nevertheless, we also provide a result for standard intersection types in the case where the right inverse is assumed to be a simple right inverse.

\medskip

The first contribution of this paper is the characterisation of the strict intersection types which can be derived for terms having left or right inverses. Building on this result we give a necessary and sufficient condition for the existence of a retraction between two strict intersection types. This condition is a generalisation of the one defined in ~\cite{BL85} for Curry types.  We show that each retraction can be witnessed by a simple right inverse. We also prove that if $\mu$ is a retract of $\nu$ and $\nu$ is a retract of $\mu$, then $\mu$ and $\nu$ are equivalent with respect to the usual subtyping relation of intersection types. Then we discuss retraction in standard intersection types. Finally we define \emph{semantic retraction} as the natural adaptation to models and we show that retraction and semantic retraction coincide. This proof uses the completeness of the filter model given in~\cite{barecoppdeza83}.

\paragraph{Outline} Section~\ref{rs} introduces the $\lambda\bot$-calculus and characterises terms having left and right inverses. The essential intersection type assignment system is defined in Section~\ref{it} together with the characterisation of the types derivable for terms with left or right inverses. The results on retraction in strict intersection types are the content of Section~\ref{ec}. In Section~\ref{git} we extend the characterisation of retractions to standard intersection types assuming that right inverses are simple.  The notion of retraction is shown equivalent to that of semantic retraction in Section~\ref{sem}. Related work is overviewed in Section~\ref{rw}. Section~\ref{c} discusses our choices and future work.

\section{Left/Right Invertible Terms}\label{rs}

Following~\cite[Definition 14.3.1]{B84}, $\Lambda \bot$ is  the set of terms obtained by adding a constant $\bot$ to the formation rules of  $\lambda$-terms. The terms of  $\Lambda \bot$ are generated by the syntax\[M::=x \mid \bot\mid\lambda x.M\mid MM\]
where $x$ ranges over a denumerable set of term variables.

The reduction rules of the $\lambda\bot$-calculus include the $\beta$-rule and two rules for $\bot$ prescribing that both application and abstraction of $\bot$ reduce to $\bot$.
\[(\lambda x.M)N\red M\set {N/x} \qquad\qquad\qquad
\bot M \red \bot\qquad\qquad\qquad
\lambda x.\bot \red \bot\]
The equality between terms is defined as $\beta\bot$-conversion, i.e. $M=N$ means that there is a term $P$ such that both $M$ and $N$ reduce to $P$.

Let ${\bf{I}}=\lambda x.x$ and  $M \circ  N = {\bf{B}}MN$, where ${\bf{B}}=\lambda xyz. x(yz)$.
We are interested in investigating the monoid of terms with the combinator ${\bf{I}}$ as identity element and  $\circ$ as binary operation.  This naturally leads to the following definition of left/right invertibility.

\begin{definition}[Left/right invertibility]\label{inv}~
  \begin{enumerate}\item A term  $M$ is \emph{left invertible} if there exists a term $L$  such that $L\circ M = \bf{I}$. We say that $L$ is a {\em left inverse} of $M$.
 \item A term  $M$ is \emph{right invertible} if there exists a term $R$  such that $M \circ R = \bf{I}$. We say that $R$ is a {\em right inverse} of $M$.
 \end{enumerate}
\end{definition}

The left/right invertibility has been studied since the seventies.
In particular the sets of terms having at least one left or one right  inverse have been characterised in~\cite{BD74} and~\cite{MZ83}.

\medskip

The characterisation of left invertible terms resorts to a set of head normal forms (hnfs for short). We recall that a hnf is a term of the shape $\lambda x_1\ldots x_n.x_j  M_1\ldots M_m$~\cite[Definition 2.2.11]{B84}.  We define a set $\Xi$ of hnfs and we show that a term is left invertible iff it reduces to a hnf belonging to $\Xi$ (Theorem~\ref{left}).

\begin {definition}\label{leftinv}
Let $\Xi$ be the set of  hnfs inductively defined as follows:
\begin{itemize}
\item $\lambda tx_1\ldots x_n.t   \in \Xi $ for $n\geq 0$;
\item if $\lambda tx_1\ldots x_n.  M_ i \in \Xi $,  then  $\lambda tx_1\ldots x_n.x_j  M_1\ldots M_i \ldots M_m \in \Xi $,  where $1 \leq  j \leq n$ and $1\leq i\leq m$.

\end {itemize}
\end {definition}

\begin{example} \label{ex0}
  ~
\begin{itemize}
\item[i)] The term $\lambda tx.xx(xt)\in\Xi $ because $\lambda tx.xt\in\Xi $, and in turn $\lambda tx.xt\in\Xi $ since $\lambda tx.t\in\Xi $.
\item[ii)]  The term $M = \lambda t x_1 x_2 . x_2 (\lambda x_3 . t)(x_1 t) \in\Xi $ since  $\lambda t x_1 x_2 x_3 .t \in \Xi$. We can also show that $M \in \Xi$ because  $\lambda t x_1 x_2 . x_1 t \in \Xi$ and in turn $\lambda t x_1 x_2 . x_1 t \in \Xi$ since $\lambda t x_1 x_2 . t \in\Xi$.
\end{itemize}
\end{example}

 Point \emph{ii)} of Example~\ref{ex0} shows that there is not, in general, a unique way to derive that a term belongs to $\Xi$. In the following, when we write $M\in\Xi $, we refer to a particular proof, chosen according to Definition~\ref{leftinv}.

In order to build a left inverse of a term in $\Xi$ we need to ``reach'' an occurrence of the first abstracted variable (called $t$ in Definition~\ref{leftinv}). We use some machinery inspired by the B\"ohm-out technique~\cite[\S 10.3]{B84}. In a hnf $\lambda x_1\ldots x_n.x_j  M_1\ldots M_m$ the number of initial abstractions is $n$,  the variable $x_j$ (bound in the $j$-th abstraction) is the head variable and $M_1,\ldots, M_m$ are the $m$ components. Following the definition of $\Xi$ we can associate with each term in $\Xi$ a list of integer triples, whose first element is the abstraction position of the head variable, whose second element is the number of components and whose third element is the position of the component used to show that the term belongs to $\Xi$ ($0$ if the component is missing). More precisely, using $^{\frown}$ to denote concatenation:

\begin{definition}\label{path} The \emph{path} $\pi(M)$ of the hnf $M\in\Xi$ is inductively defined by:
\begin{itemize}
 \item $\pi(\lambda tx_1\ldots x_n.t)=\trp 1 0 0$;
\item $\pi(\lambda tx_1\ldots x_n.x_j  M_1 \ldots M_m)=\trp {j+1} m i^{\frown}\pi(\lambda tx_1\ldots x_n.M_i)$ if $\lambda tx_1\ldots x_n.  M_ i \in \Xi$ is used to show $\lambda tx_1\ldots x_n.x_j  M_1 \ldots M_m\in\Xi$.
\end{itemize}
\end{definition}
\noindent Let $\p$ range over paths.
\begin{example}
We get $\pi(\lambda tx.xx(xt))=\trp222^{\frown}\trp211^{\frown}\trp100$. If $M$ is defined as in Point ii) of Example~\ref{ex0} we get either $\pi(M)=\trp321^{\frown}\trp100$ or $\pi(M)=\trp322^{\frown}\trp211^{\frown}\trp100$, according to the proof used to show $M\in\Xi$.
\end{example}
The triple $\trp {j+1} m i$ says that the variable bound in the $j+1$-th abstraction must choose the $i$-th component out of $m$ components to ``reach'' an occurrence of $t$. Then the variable bound in the $j+1$-th abstraction  needs to be replaced by the term $S_i^{(m)}=\lambda y_1\ldots y_m.y_i$. We call {\em selectors} the terms of the shown shape. This replacement becomes problematic if we have in the same path two triples with the same first element which differ in one of the other elements. Following~\cite{MZ83} we differentiate these occurrences using terms of the shape $P^{(m)}=\lambda z_1\ldots z_{m+1}.
z_{m+1} z_1 \ldots z_m$ with $m\geq 1$ that we dub {\em permutators}.  We convene that $P^{(0)}=\bot$. We need two preliminary definitions and a technical lemma.

By $\ell(\p)$ we denote the length, i.e. the number of triples, of the path $\p$.

We define $\#(j+1,\p)$ as the maximum of the second components of triples in path $\p$ whose first component is $j+1$. We assume $\#(j+1,\p)=0$ if $j+1$ does not occur in $\p$ as first component. More formally:
\[ \#(j+1,\trp 1 0 0)=0\qquad\qquad\qquad\#(j+1,\trp h m i^{\frown}\p)=\begin{cases}
\text{max}(m, \#(j+1,\p))     & \text{if }h=j+1, \\
  \#(j+1,\p)    & \text{otherwise}
\end{cases}\]

\begin{lemma}\label{aux}
Let $M=\lambda tx_1\ldots x_n.M'\in\Xi$ and $m_j\geq\#(j+1,\pi(M))$ for $1\leq j\leq n$. Then \[\lambda t.M'\set{P^{(m_1)}/x_1}\ldots \set{P^{(m_n)}/x_n}=Q\in \Xi\] and $\ell(\pi(Q))=\ell(\pi(M))$.
\end{lemma}
\begin{proof}
The proof is by induction on the definition of $\Xi$.
If $M'=\lambda x_{n+1}\ldots x_q.t$, then $Q=\lambda tx_{n+1}\ldots x_q.t$.

\noindent
Let  $M'=\lambda x_{n+1}\ldots x_q.x_r M_1\ldots M_i\ldots M_m$ and $M\in\Xi$  since $  \lambda tx_1\ldots x_q.  M_ i\in\Xi$. Because $m_j\geq\#(j+1,\pi(M))$ implies $m_j\geq\#(j+1,\pi( \lambda tx_1\ldots x_q.  M_ i))$ for $1\leq j\leq n$, by induction we get \[\lambda tx_{n+1}\ldots x_q.  M_ i\sub=\lambda tx_{n+1}\ldots x_q.  Q'\in\Xi\] where $\sub$ is the substitution $\set{P^{(m_1)}/x_1}\ldots \set{P^{(m_n)}/x_n}$. If $r\geq n+1$ we can take \[Q=\lambda t x_{n+1}\ldots x_q.x_r M_1\sub\ldots M_{i-1}\sub Q'M_{i+1}\sub\ldots M_m\sub\]
Otherwise

$\begin{array}{lll}
M'\sub&=&\lambda t x_{n+1}\ldots x_q. P^{(m_r)}M_1\sub\ldots M_{i-1}\sub Q'M_{i+1}\sub\ldots M_m\sub\\
&= &\lambda t x_{n+1}\ldots x_qz_{m+1}\ldots z_{m_r+1}.z_{m_r+1}M_1\sub\ldots M_{i-1}\sub Q'M_{i+1}\sub\ldots M_m\sub z_{m+1}\ldots z_{m_r}\end{array}$\\
and we can take this last hnf as $Q$ since $\lambda tx_{n+1}\ldots x_q.  Q'\in\Xi$ implies $\lambda t x_{n+1}\ldots x_qz_{m+1}\ldots z_{m_r+1}.Q'\in\Xi$.

In all cases it is easy to verify that $\ell(\pi(Q))=\ell(\pi(M))$.
\end{proof}

\begin{example}\label{ex1}
Let $M=\lambda tx.M'$ where $M'= xx(xt)$. We get
\[\lambda t. M'\set{P^{(2)}/x}=\lambda t z_1.z_1(\lambda z_2z_3z_4. z_4 z_2 z_3) (\lambda z_5z_6. z_6 t z_5)=Q\] We have $\pi(Q)=\trp222^{\frown}\trp421^{\frown}\trp100$. \end{example}

\medskip

\begin{theorem}\label{left}
A term has at least one left inverse if and only if it reduces to a hnf  $M$ in $\Xi$.
\end{theorem}
\begin{proof}~
  \begin{description}
    \item[(If)]
      Let $M\in\Xi$.
      The proof is by induction on $\ell(\pi(M))$.
      Let $n$ be the number of initial abstractions and $\p$ be  the path of the term considered in the induction step.
      We build  a  left inverse of the shape $\lambda z. z L_1\ldots L_q$, where $q\geq n$ and $L_l=\bot$ whenever $\#(l+1,\p)=0$.

      If  $M=\lambda tx_1\ldots x_n.t$,  then $\lambda z.z\underbrace{\bot \bot \cdots \bot}_n$ is a left inverse of $M$.

      Let  $M=\lambda tx_1\ldots x_n.x_j M_1\ldots M_i\ldots M_m\in\Xi$  since $  \lambda tx_1\ldots x_n.  M_ i\in\Xi$. We distinguish two cases. In the first case the construction of the left inverse using a selector is easy. In the second case we compose $M$ with a term $N$ build out of permutators. The useful property is that $N\circ M=Q\in\Xi$ and  $\ell(\pi(Q))=\ell(\pi(M))$ and $Q$ satisfies the condition of case 1. We can then build a left inverse $L$ of $Q$ and $L\circ N$ is a left inverse of $M$.\\
      {\em Case 1}:  $\#(j+1,\pi(\lambda tx_1\ldots x_n.  M_ i))=0$.
      By induction hypothesis $  \lambda tx_1\ldots x_n.  M_ i\in\Xi$ has a left inverse $\lambda z. z L'_1\ldots L'_q$ and in this case $L'_j
      =\bot$. Then $\lambda z. z L'_1\ldots L'_{j-1}S^{(m)}_i L'_{j+1}\ldots L'_q$ is a left inverse of $M$.\\
      {\em Case 2}:  $\#(j+1,\pi(\lambda tx_1\ldots x_n.  M_ i))=m_j\not=0$. Let $m_l=\#(l+1,\pi(M))$ for $1\leq l\leq j$ and \[N=\lambda z.z P^{(m_1)}\ldots P^{(m_j)}\] By the proof of Lemma~\ref{aux} $N\circ M=Q\in\Xi$, where $Q$ is the hnf \[\lambda t x_{j+1}\ldots x_nz_{m+1}\ldots z_{m_j+1}.z_{m_j+1}M_{1}\sub\ldots M_m\sub z_{m+1}\ldots z_{m_j}\] and $\sub$ is the substitution $\set{P^{(m_1)}/x_1}\ldots \set{P^{(m_j)}/x_j}$. Since $z_{m_j+1}$ does not occur in $M_i\sub$, i.e. \[\#(n-j+m_j-m+2,\pi(\lambda tx_{j+1}\ldots x_nz_{m+1}\ldots z_{m_j+1}.  M_ i\sub))=0\]
      and $\ell(\pi(Q))=\ell(\pi(M))$ we can build a left inverse $L$ of $Q$ according to previous case. Then a left inverse of $M$ is $L\circ N$.


    \item[(Only if)] Let us suppose, ad absurdum, that a term has a left inverse and it is unsolvable or its hnf  doesn't belong to $\Xi$.  The first case is obvious. In the second case the hnf $M$ has  no path which satisfies Definition~\ref{path}. Therefore, if $M=\lambda t.N$, then there is no occurrence of $t$ in $N$ which is not applied and such that it is always in components whose head variables are bound. The arguments of $t$ cannot be erased by reduction and a free variable cannot be replaced in order to get $t$. So we conclude that $M$ has no left inverse.
      \qedhere
  \end{description}
\end{proof}

\begin{example}\label{ex2}
  Let $M, Q$ be as in Example~\ref{ex1}. The left inverse of  $Q$ built according to the lemma is $L=\lambda z. z (\lambda y_1y_2. y_2)\bot(\lambda y_1y_2. y_1)$. According to the proof of previous theorem we get $N=\lambda z. z P^{(2)}$. Then a left inverse of $M$ is  $L\circ N=\lambda z. z P^{(2)}(\lambda y_1y_2. y_2)\bot(\lambda y_1y_2. y_1)$.\end{example}

\medskip

The characterisations of terms having right inverses is easy.

\begin{theorem}\label{right}
  A term has at least one right inverse if and only if its hnf is of the shape: $\lambda z.zM_1 \ldots  M_m$.
\end{theorem}
\begin{proof}~
  \begin{description}
    \item[(If)]
      A right inverse is $\lambda t x_1 \ldots x_m.t$.

    \item[(Only if)] An unsolvable term has no left inverse.  Let  suppose the hnf of a term  be not of the shape  $\lambda z.zM_1 \ldots  M_m$. Then it must have more than one abstraction and/or the head variable must be a free variable. In the first case the initial abstractions and in the second case the head free variable cannot be eliminated using reductions.
      \qedhere
  \end{description}
\end{proof}

\begin{example}\label{ex3}
  The term $M$ of Example~\ref{ex1} is a right inverse of the term  $L \circ N$ of Example~\ref{ex2}. The right inverse of $L \circ N$ built by the theorem is $\lambda t x_1x_2 x_3.t$.
\end{example}

From the proof of Theorem~\ref{right}  it is clear that if a term has a right inverse, then it has also a right inverse of the shape $\lambda t x_1 \ldots x_n.t$, i.e. a selector $S_1^{(n+1)}$. We call \emph{simple right inverses} the hnfs of this shape.

\section{Strict  Intersection Types}\label{it}

The type system considered in this paper is a notational variant of the essential  intersection assignment  introduced in~\cite{B11}.

\indent The set of {\em strict intersection types} is defined by:
\begin{center}
$\begin{array}{lll}
\mu&:=&\varphi~~\mid~\omega~~\mid ~\sigma \to\mu\\
\tA&:=&\mu~\mid~\sigma\wedge\sigma
\end{array}$
\end{center}

\noindent
where $\varphi$  ranges over type variables and $\omega$ is a constant.
We convene that  $\mu,\nu$ range over strict intersection types (either atomic or arrow types), while $\sigma , \tau  , \rho$   range over intersections.

\noindent
Conventionally, we omit parentheses according to the precedence rule ``$\wedge$ over $\rightarrow$" and we assume that $\to$ associates to the right.
Intersections are considered  modulo idempotence, commutativity and associativity  of $\wedge$. In this section and in the following one we use type as short for strict intersection type.

\smallskip

A preorder relation $\leq$, 
representing set inclusion, is assumed between  types and intersections.

\begin{definition}\label{tam}
Let $\leq$ be the minimal reflexive and transitive  relation such that:
$$
\begin{array} {c}
\sigma\leq\omega\qquad\qquad \omega \leq \omega \to \omega \qquad\qquad
 \sigma_1 \wedge \sigma_2 ~\leq~ \sigma_i\quad
~~~(i=1,2) \\[.5em]
\sigma_1 \leq \tau_1\text{ and }  \sigma_2\leq \tau_2~\text{imply} ~ \sigma_1\wedge\sigma_2 ~\leq~ \tau_1 \wedge \tau_2 \\[.5em]
\sigma_2 \leq \sigma_1 \text{ and }  \mu_1 \leq \mu_2 ~\text{imply} ~ \sigma_1 \to \mu_1 \leq  \sigma_2 \to  \mu_2
 \end{array} $$
\end{definition}
\noindent
 We write $\sigma\sim\tau$ if $\sigma\leq\tau$ and $\sigma\geq\tau$.

 A key property of this subtyping is the content of the following lemma, for a proof see~\cite{B11}.
\begin{lemma}\labelx{2.4(i)}
If $\sigma\to\mu\leq \tau \to \nu$, then $\tau\leq\sigma$ and $\mu\leq\nu$.
\end{lemma}

The essential intersection type assignment system is defined by the typing rules of Table \ref{tr}. We assume that an environment associates intersections with a finite number of term variables. Let $\Gamma$ range over environments. The subsumption rule uses the preorder of Definition~\ref{tam}. We write $\B\vdash N:\sigma$ with $\sigma=\bigwedge_{i\in I} \mu_i$ as short for  \mbox{$\B\vdash N:\mu_i$} for all $i\in I$.

\begin{table}[h]
$$
\begin{array}{c}
(Ax)\quad  \B, x: \bigwedge_{i\in I} \mu_i \vdash x: \mu_j\quad j\in I \qquad\qquad(\omega) \quad \B \vdash M: \omega
\qquad\qquad (\leq) \quad  \db \frac{\B \vdash M: \mu \quad \mu \leq \nu}
{\B \vdash M:  \nu}
\\\\
(\to I) \quad  \db \frac{\B,x:\tA \vdash M: \mu}
{\B \vdash \lambda x.M: \tA \to \mu} \qquad\qquad
(\to E) \quad \db \frac{\B \vdash M: \tA \to \mu \quad
\B \vdash N: \tA}{\B \vdash MN: \mu}
\end{array}
$$
\caption{Typing Rules}\label{tr}
\end{table}

\smallskip

The inversion lemma is as expected, for a proof see~\cite{B11}.

\begin{lemma}[Inversion Lemma]\label{il} 
  ~
\begin{enumerate}
\item\label{il1} If $\Gamma\vdash x:\mu$, then either $\mu\sim\omega$ or $x:\sigma\in\Gamma$ and $\sigma\leq \mu$.
\item \label{il4} If $\Gamma\vdash \bot:\mu$, then $\mu\sim\omega$.
\item \label{il2} If $\Gamma\vdash MN:\mu$, then $\Gamma\vdash M:\sigma\to\mu$ and $\Gamma\vdash N:\sigma$.
\item \label{il3} If $\Gamma\vdash \lambda x.M:\mu$, then $\mu\sim\sigma\to\nu$ and $\Gamma,x:\sigma\vdash M:\nu$.
\end{enumerate}
\end{lemma}
In this system types are  preserved by $\beta\bot$-conversion~\cite{B11}:

\begin{theorem}[Subject Conversion]\label{sr}
If $\B\vdash M:\mu$ and $M=N$, then $\B\vdash N:\mu$.
\end{theorem}

We say that {\em a term inhabits a type} if we can derive the type for the term starting  from the empty environment.  A term inhabits a set of types if it inhabits all the types belonging to the set. Let $\tau =\bigwedge_{i\in I} \mu_i$: we say that $\sigma\to\tau$ is inhabited if there exists a term which inhabits all the types $\sigma\to\mu_i$ for ${i\in I}$.

Inhabitation for intersection types has been shown undecidable in general~\cite{U99}, but decidable for types with {\em rank} less than or equal to 2~\cite{U09}, when the rank of  types and intersections is defined by:
\[\begin{array}{lll}
\rank(\mu)&=&\begin{cases}
   \max(\rank(\sigma)+1,\rank(\nu))   & \text{if }\mu=\sigma\to\nu\text{ and }\wedge\text{ occurs in }\mu, \\
   0   & \text{otherwise}
\end{cases}\\[20pt]
\rank(\sigma\wedge\tau)&=&\max(1,\rank(\sigma),\rank(\tau))
\end{array}\]

In the following we characterise the types of left/right invertible terms. These characterisations require inhabitation of some types and therefore they are effective only when these types are of rank at most 2.

\smallskip

We start defining inductively the set $\Theta$ of left types which mimic the set $\Xi$ of hnfs, that is the construction of $\Theta$ follows the construction of $\Xi$.  In the following definition $\sigma\to\sigma_1\to\ldots\to\sigma_n\to\tau\in\Theta$, where  $\tau=\bigwedge_{i\in I} \mu_i$, is used as short for $\sigma\to\sigma_1\to\ldots\to\sigma_n\to\mu_i\in\Theta$ for all $i\in I$.

\begin{definition}
The set $\Theta$ of {\em left types} is inductively defined by:
\begin{itemize}
\item if $\sigma\leq\nu$, then $\sigma\to\sigma_1\to\ldots\to\sigma_n\to\nu\in \Theta$;
\item if $\sigma\to\sigma_1\to\ldots\to\sigma_n\to\rho_i\in \Theta$ and $\sigma_j\leq\rho_1\to\ldots\to\rho_m\to\nu$ for some $j\leq n$ and $i\leq m$ and \mbox{$\sigma\to\sigma_1\to\ldots\to\sigma_n\to\rho_l$} is inhabited for $1\leq l\leq m$, then $\sigma\to\sigma_1\to\ldots\to\sigma_n\to\nu\in \Theta$.
\end{itemize}
\end{definition}

\needspace{1em}
\begin{example}\label{ex4}
  ~
\begin{itemize}
\item[i)] Let $\tau=\psi\wedge(\varphi\to\varphi')\wedge(\psi\to\varphi'\to\psi')$. We have $\varphi\to\tau\to\psi'\in\Theta$ since:
\begin{itemize} \item $\varphi\to\tau\to\varphi'\in\Theta$
\item $\tau \leq \psi\to\varphi'\to\psi'$ and
\item both $\varphi\to\tau\to\psi$ and $\varphi\to\tau\to\varphi'$ are inhabited. \end{itemize}
Moreover $\varphi\to\tau\to\varphi'\in\Theta$ since:
\begin{itemize} \item $\varphi\to\tau\to\varphi\in\Theta$
\item $\tau \leq \varphi\to\varphi'$ and
\item $\varphi\to\tau\to\varphi$ is inhabited. \end{itemize}
The type $\varphi\to\tau\to\psi'$ can be derived for the term $M$ of Example~\ref{ex1}.
\item[ii)]  Let $\mu=(\varphi\to\psi)\to\psi\to\psi'$ and $\nu=\varphi\to\psi$ . We have $\varphi\to\mu\to\nu\to\psi'\in\Theta$ since:
\begin{itemize} \item $\varphi\to\mu\to\nu\to\psi\in\Theta$
\item $\mu\leq \mu$ and
\item both $\varphi\to\mu\to\nu\to\nu$ and $\varphi\to\mu\to\nu\to\psi$ are inhabited. \end{itemize}
Moreover $\varphi\to\mu\to\nu\to\psi\in\Theta$ since:
 \begin{itemize} \item $\varphi\to\mu\to\nu\to\varphi\in\Theta$
\item $\nu\leq \nu$ and
\item $\varphi\to\mu\to\nu\to\varphi$ is inhabited. \end{itemize}
 The type $\varphi\to\mu\to\nu\to\psi'$ can be derived for the left invertible term $\lambda t x_1 x_2.x_1 x_2 (x_2 t)$.
\end{itemize}
\end{example}
\smallskip

We define the number of top arrows of a type as expected:
\[\flat(\varphi)=\flat(\omega)=0 \qquad \qquad \flat(\sigma\to\mu)=1+\flat(\mu)\]
It is useful to observe  that if a type with at least $n$ top arrows has an inhabitant, then this type has also an inhabitant with at least $n$ initial abstraction.

\begin{lemma}\label{taia}
If type $\mu$ is inhabited and $\flat(\mu)\geq n$, then there is $M$ with at least $n$ initial abstractions such that $\vdash M:\mu$.
\end{lemma}
\begin{proof} If $\mu\sim \omega$ it is trivial. Otherwise the inhabitants of $\mu$ must have hnfs, see~\cite{B11} for a proof. Let $\mu=\sigma_1\to \ldots\to\sigma_{n}\to\nu$ and $\lambda x_1\ldots x_{n'}.x_jM_1\ldots M_m$ be an inhabitant of $\mu$ with
 $n'< n$. It is easy to check that we get $M$ by $\eta$-expansion (see~\cite[Definition 3.3.1]{B84}) \[\vdash \lambda x_1\ldots x_{n'}y_1\ldots y_{n-n'}.x_jM_1\ldots M_my_1\ldots y_{n-n'}: \sigma_1\to \ldots\to\sigma_{n}\to\nu\qedhere
 \]
\end{proof}

\begin{lemma}[Characterisation of Types for Left Invertible Terms]\label{ctlit}~
\begin{enumerate}
\item\label{ctlit1} Left invertible terms inhabit only types which are
left types.
\item\label{ctlit2} Each left type is inhabited by a left invertible term.
\end{enumerate}
\end{lemma}
\needspace{2em}
\begin{proof}~
  \begin{enumerate}
    \item 
      By Definition~\ref{leftinv} a left invertible term $M$ is a $\lambda$-abstraction,  then its type is of the shape $\sigma\to\nu$ by Lemma~\ref{il}(\ref{il3}).
      It is enough to show that  if the head normal form of $M$ belongs to $\Xi$, then $\sigma\to\nu\in \Theta$. The proof is by induction on $\Xi$.

      Case $\lambda tx_1\ldots x_n.t$. By the invariance of types under $\beta\bot$-conversion (Theorem~\ref{sr}) we get \[\vdash\lambda tx_1\ldots x_n.t:\sigma\to\nu\]which implies by repeated application of Lemma~\ref{il}(\ref{il3}) $\nu\sim\sigma_1\to\ldots\to\sigma_n\to\mu$ and \[t:\sigma,x_1:\sigma_1, \ldots, x_n:\sigma_n\vdash t:\mu\]
      Then $\sigma\leq\mu$ by Lemma~\ref{il}(\ref{il1})  and we can conclude $\sigma\to\nu\in \Theta$.

      Case $\lambda tx_1\ldots x_n.x_j  M_1\ldots M_m \in \Xi $ with $j\leq n$  since  $\lambda tx_1\ldots x_n.  M_ i \in \Xi $ with $i\leq m$. As in previous case we get $\nu\sim\sigma_1\to\ldots\to\sigma_n\to\mu$ and
      \[t:\sigma,x_1:\sigma_1, \ldots, x_n:\sigma_n\vdash x_j  M_1\ldots M_m:\mu\] Let $\B=t:\sigma,x_1:\sigma_1, \ldots, x_n:\sigma_n$.
      By repeated application of Lemma~\ref{il}(\ref{il2}) we have \begin{center}$\B\vdash x_j:\rho_1\to\ldots\rho_m\to\mu$ and $\B\vdash M_l:\rho_l$ for $1\leq l\leq m$.\end{center} Lemma~\ref{il}(\ref{il1}) implies $\sigma_j\leq\rho_1\to\ldots\rho_m\to\mu$. Moreover\begin{center}$\lambda tx_1\ldots x_n.M_l$ inhabits $\sigma\to\sigma_1\to\ldots\to\sigma_n\to\rho_l$ for $1\leq l\leq m$.\end{center} Lastly $\lambda tx_1\ldots x_n.  M_ i \in \Xi $ implies by induction \[\sigma\to\sigma_1\to\ldots\to\sigma_n\to\rho_i\in\Theta\] We can then conclude $\sigma\to\nu\in \Theta$.

    \item 
      The proof is by induction on $\Theta$. If $\sigma\leq \nu$ we can derive \mbox{$\vdash\lambda tx_1\ldots x_n.t: \sigma\to \sigma_1\to \ldots\to\sigma_n\to\nu$.} Otherwise by Lemma~\ref{taia} we can assume that the inhabitants of $\sigma\to\sigma_1\to \ldots\to\sigma_n\to\rho_l$ for $1\leq l\leq m$ have at least $n+1$ initial abstractions. Let \begin{center}$\lambda tx_1\ldots x_n.M_l$ be an inhabitant of $\sigma\to\sigma_1\to \ldots\to\sigma_n\to\rho_l$ for $1\leq l\leq m$ and $\sigma\to\sigma_1\to \ldots\to\sigma_n\to\rho_i\in\Theta$ for some $i\leq m$.\end{center} By induction $\lambda tx_1\ldots x_n.M_i\in\Xi$, then $\lambda tx_1\ldots x_n.x_j M_1\ldots M_i\ldots M_m\in\Xi$. Moreover \begin{center}if $\sigma_j\leq \rho_1\to\ldots\to\rho_m\to\nu$,  then $\lambda tx_1\ldots x_n.x_j M_1\ldots M_i\ldots M_m$ inhabits $\sigma\to\sigma_1\to \ldots\to\sigma_n\to\nu$.\end{center}
      \qedhere
  \end{enumerate}
\end{proof}

The types of right invertible terms are easy to define, as expected.

\begin{definition}\label{rt}
  A type $\tau \to \mu$ is a {\em right type} if $\tau\leq \rho_1\to\ldots\to\rho_m\to\mu$ and $\tau\to \rho_i$ is inhabited for $1\leq i\leq m$.
\end{definition}

\begin{example}\label{ex5} A right type is $((\psi_1\to\psi_1)\wedge(\psi_2\to\psi_2)\to\omega\to\varphi)\wedge\psi\to\varphi$.  Another right type is
  \[((\varphi_1\to\varphi_2\to(\varphi_1\to\varphi_2\to\varphi_3)\to\varphi_3)\to(\psi_1\to\psi_2\to\psi_2)\to\omega\to(\psi_1\to\psi_2\to\psi_1)\to\varphi)\to\varphi\]
  This last type can be derived for
  the term $L\circ N$ of Example~\ref{ex2}.
\end{example}

\begin{lemma}[Characterisation of Types for Right Invertible Terms]\label{ctrit}~
  \begin{enumerate}
    \item\label{ctrit1} Right invertible terms inhabit only types which are 
      right types.
    \item \label{ctrit2} Each right type is inhabited by a right invertible term.
  \end{enumerate}
\end{lemma}
\begin{proof}~
  \begin{enumerate}
    \item 
      By Theorem~\ref{right} is it enough to show that $\vdash \lambda z.zM_1\ldots  M_m:\tau \to \mu$ implies that $\tau \to \mu$ is a right type. By Lemma~\ref{il}(\ref{il3}) and (\ref{il2}) we get:
      $z:\tau\vdash z:\rho_1\to\ldots\to\rho_m\to\mu$ and $z:\tau\vdash M_i:\rho_i$ for $1\leq i\leq m$. By Lemma~\ref{il}(\ref{il1}) $\tau\leq \rho_1\to\ldots\to\rho_m\to\mu$. Moreover $\lambda z.zM_i$ inhabits $\tau\to\rho_i$ for $1\leq i\leq m$. Therefore $\tau \to \mu$ is a right type.

    \item 
      Let $M_i$ be an inhabitant of $\tau\to \rho_i$ for $1\leq i\leq m$ and $\tau\leq \rho_1\to\ldots\to\rho_m\to\mu$. Then we can derive $\vdash \lambda z.zM_1\ldots  M_m:\tau \to \mu$.
      \qedhere
  \end{enumerate}
\end{proof}

It is easy to verify that $\omega$ is both a left and a right type. 
\section{Characterisation of Retraction in Strict Intersection Types}\label{ec}
We can discuss now retractions, i.e. isomorphic embeddings, in strict types using terms of $\Lambda\bot$.

\begin{definition} \label{defcoer}
 Type $\mu$ is a \emph{retract} of type $\nu$   (notation $\mu \lhd \nu$) if there exist terms $L$ and $R$ such that:
\begin{enumerate}
\item \label{defcoer11} $\vdash L: \nu \to \mu$;
\item \label{defcoer12} $\vdash R: \mu \to \nu$;
\item \label{defcoer13} $L \circ R = {\bf I }$.
\end{enumerate}
We say that $L,R$ witness the retraction.
\end{definition}

\begin{example}\label{ex6}
$L=\lambda z. z{\bf I}(\lambda y.yy)z$ and $R=\lambda t x_1x_2x_3.x_2x_1t$ witness the retraction\[\varphi\lhd (\varphi\to\varphi)\wedge((\varphi\to\varphi)\to\varphi\to\varphi)\to\sigma\to\omega\to\varphi\] where $\sigma=(\varphi\to\varphi)\wedge((\varphi\to\varphi)\to\varphi\to\varphi)\to\varphi\to\varphi$. The same retraction is witnessed by $L$ and $R'=\lambda t x_1x_2x_3.t$, which is a simple right inverse. Notice that $L$ cannot be typed with Curry types.
\end{example}

It easy to prove that the retraction relation enjoys the transitivity property. In fact if $L,R$ witness  $\mu \lhd \mu'$ and $L',R'$ witness  $\mu' \lhd \nu$, then $L\circ L',R'\circ R$  witness  $\mu \lhd \nu$.

\medskip

Retraction can be fully characterised.
\begin{theorem}[Characterisation of  Retraction]\label{cec}
$\mu\lhd \nu$  if and only if
$\nu \sim \rho_1 \to  \ldots\to\rho_m \to \mu$ and $\nu\to\rho_i$ is inhabited for $1\leq i\leq m$. Moreover, each retraction can be witnessed by a simple right inverse.

 \end{theorem}
\begin{proof}~
  \begin{description}
    \item[(If)]
      Let $M_i$ be an inhabitant of $\nu\to\rho_i$ for $1\leq i\leq m$. We can choose $L=\lambda z.z(M_1z)\ldots (M_mz)$ and $R=\lambda tx_1\ldots x_m.t$. Notice that $R$ is a simple right inverse. It it easy to verify that $L$ and $R$ satisfy the conditions of Definition~\ref{defcoer}.
    \item[(Only if)]
      By Theorem~\ref{right} $L=\lambda z.zM_1 \ldots  M_m$. Then applying Lemma~\ref{il}(\ref{il3}) to $\vdash L: \nu \to \mu$ we get \mbox{$ z: \nu \vdash zM_1 \ldots  M_m : \mu$.} By repeated applications of Lemma~\ref{il}(\ref{il2}) this implies\\  \centerline{$ z: \nu \vdash z:\sigma_1 \to  \ldots\to\sigma_m \to \mu$  and $ z: \nu \vdash M_i : \sigma_i$ for $1\leq i \leq m$.} From $ z: \nu \vdash z:\sigma_1 \to  \ldots\to\sigma_m \to \mu$ we have
      $\nu \leq \sigma_1 \to  \ldots\to\sigma_m \to \mu$ by Lemma~\ref{il}(\ref{il1}).
      We can assume $\nu = \rho_1 \to  \ldots\to\rho_m \to \mu' $, which implies $\sigma_i \leq \rho_i$ for $1 \leq i \leq  m$ and $\mu' \leq\mu$ by Lemma~\ref{2.4(i)}. Observe that $L \circ R=\l x. RxM'_1 \ldots M'_m$  where $M_i'=M_i\set {Rx/z}$ for $1\leq i\leq m$. From $\vdash R: \mu \to \nu$ and $ z: \nu \vdash M_i : \sigma_i$ and $\sigma_i \leq \rho_i$ we can derive
      $ x: \mu \vdash M_i' : \rho_i$ for $1\leq i\leq m$. This together with $\nu = \rho_1 \to  \ldots\to\rho_m \to \mu' $ implies $ x: \mu \vdash RxM'_1 \ldots M'_m :\mu'$. From $L \circ R = {\bf I }$ we get $RxM'_1 \ldots M'_m=x$. Subject Conversion derives $ x: \mu \vdash x:\mu'$, so by Lemma~\ref{il}(\ref{il1}) $\mu\leq\mu'$. We conclude $\mu\sim\mu'$.
      \qedhere
  \end{description}
\end{proof}

As an easy consequence of this theorem  if $\mu \lhd \nu$, then $\mu \lhd \rho \to \nu$ for any intersection $\rho$ such that $\nu\to\rho$ is inhabited. Moreover if $\mu\not\sim\omega$, then neither $\mu \lhd \omega$ nor $\omega \lhd \mu$ can hold.

\medskip

Reciprocal retraction implies equivalence.
\begin{corollary}
  If $\mu\lhd \nu$ and $\nu\lhd \mu$, then $\mu\sim \nu$.
\end{corollary}
\begin{proof} By previous theorem $\nu \sim \rho_1 \to  \ldots\to\rho_m \to \mu$ and $\mu \sim \sigma_1 \to  \ldots\to\sigma_n \to \nu$, which imply $m=n=0$ and then $\mu\sim \nu$.\end{proof}

\medskip
Given two terms $L$ and $R$ such that $L \circ R = {\bf I }$, they do not  witness a retraction for all the pairs of their types.
For instance,  the terms $L \circ N$, $M$ in Examples \ref{ex2} and \ref{ex1}
do not witness a retraction using their types shown in Examples \ref{ex5} and \ref{ex4}\,{\it i)}.
However  $L \circ N$, $M$ witness the  following retraction:
\[\mu~~\lhd~~(\mu\to\nu_1)\wedge(\omega\to\nu_1\to\nu_2)\to\nu_2\]
where $\nu_1 = \omega\to(\mu\to\omega\to\mu)\to\mu$ and $\nu_2 = (\omega\to\nu_1\to\nu_1)\to\nu_1$.

\medskip

Notice that $L\circ R={\bf I}$ implies that for all types $\mu$ there is an intersection $\sigma$ such that $L$ has type $\sigma\to\mu$ and $R$ has type $\mu\to\sigma$. The proof is easy using the inversion lemma and the invariance of types under $\beta$-expansion. This does not mean that $L,R$ always witness a retraction between strict types, since $\sigma$ can be an intersection of strict types, as shown in the following example.
\begin{example}\label{ines}Let $L = \lambda u.u\bot(u\bot(\lambda z.\textbf{I}))$ and $R = \lambda t x_1 x_2. x_2 t$. We derive \[\vdash L: (\omega\to(\varphi\to\varphi)\to\varphi)\wedge(\omega\to(\varphi\to\varphi\to\varphi)\to\varphi\to\varphi)\to\varphi\]
  \[\vdash R: \varphi\to\omega\to(\varphi\to\varphi)\to\varphi\text{ and }\vdash R: \varphi\to\omega\to(\varphi\to\varphi\to\varphi)\to\varphi\to\varphi\]
\end{example}


\noindent
When $R$ is a simple right inverse instead we always get a set of retractions.
\begin{theorem}
  If $L\circ R={\bf I}$ and $R$ is a simple right inverse, then for each type $\mu$ we find types $\nu$ such that $L,R$ witness the retraction $\mu\lhd\nu$.
\end{theorem}
\begin{proof} By Theorem~\ref{right} $L=\lambda z.zM_1 \ldots  M_m$. By definition of simple right inverse $R=\lambda tx_1\ldots x_m.t$. We can then choose $\nu\sim\rho_1\to\ldots\to\rho_m\to\mu$ for each $\rho_1,\ldots,\rho_m$ such that $\vdash \lambda z.M_i:\nu\to\rho_i$ for $1\leq i\leq m$. In particular we can always take $\rho_i=\omega$ for $1\leq i\leq m$.
\end{proof}

\medskip

\section{Characterisation of Simple Retraction in Standard Intersection Types}\label{git}
In this section we extend the characterisation of retraction (Theorem~\ref{cec}), which holds for strict types, to the case of standard intersection types, by considering only simple right inverses.

The set of standard intersection types (simply {\em intersection types}) is defined by:
\[
  \sigma := \varphi~\mid~\omega~\mid~\sigma\to\sigma~\mid~\sigma\wedge\sigma
\]
where, as before, $\varphi$  ranges over type variables and $\omega$ is a constant.
In this section, we convene that $\sigma , \tau, \rho$   range over intersection types.

The preorder of Definition~\ref{tam} is extended to intersection types by adding the rule
\[(\sigma\to\tau)\wedge(\sigma\to\rho)\leq\sigma\to\tau\wedge\rho\]
The following lemma gives a crucial property of this subtyping, which is shown in~\cite{barecoppdeza83}.
\begin{lemma}\labelx{2.4(ii)}
If $\bigwedge_{i\in I} (\sigma_i\to\tau_i)\leq \sigma\to\tau$, then there is $J\subseteq I$ such that $\sigma\leq \bigwedge_{i\in J}\sigma_i$ and $\bigwedge_{i\in J}\tau_i\leq\tau$.
\end{lemma}

\medskip

The type assignment system is defined by the typing rules of Table~\ref{trgit}, where environments are finite mappings from term variables to intersection types.

\begin{table}[h]
  $$
  \begin{array}{c}
    (Ax)\quad  \B, x: \sigma\vdash x:\sigma \qquad\qquad(\omega) \quad \B \vdash M: \omega
    \qquad\qquad (\leq) \quad  \db \frac{\B \vdash M: \sigma \quad \sigma \leq \tau}
    {\B \vdash M:  \tau}
    \\\\
    (\to I) \quad  \db \frac{\B,x:\tA \vdash M: \tau}
    {\B \vdash \lambda x.M: \tA \to \tau}
    \qquad
    (\to E) \quad \db \frac{\B \vdash M: \tA \to \tau \quad
    \B \vdash N: \tA}{\B \vdash MN: \tau}
    \qquad
    (\wedge I) \quad \db \frac{\B\vdash M:\tA\quad\B\vdash M:\tau}
    {\B\vdash M:\tA\wedge\tau}
  \end{array}
  $$
  \caption{Typing Rules of Standard Intersection Types}\label{trgit}
\end{table}

\medskip

The inversion lemma is as expected, for a proof see~\cite[Theorem 12.1.13]{BDS13}.
\begin{lemma}[Inversion Lemma]\labelx{ils}~
  \begin{enumerate}
    \item\labelx{ils1} If $\Gamma\vdash x:\tau$, then either $\tau\sim\omega$ or $x:\sigma\in \Gamma$ for $\sigma\leq\tau$.
    \item \label{ils4} If $\Gamma\vdash \bot:\tau$, then $\tau\sim\omega$.
    \item\labelx{ils2} If $\Gamma\vdash MN:\tau$, then $\Gamma\vdash M:\sigma\to\tau$ and $\Gamma\vdash N:\sigma$.
    \item\labelx{ils3} If $\Gamma\vdash\lambda x.M:\tau$, then $\tau\sim\bigwedge_{i\in I}(\sigma_i\to\rho_i)$ and $\Gamma,x:\sigma_i\vdash M:\rho_i$ for all $i\in I$.
  \end{enumerate}
\end{lemma}

\medskip

The notion of retraction (Definition~\ref{defcoer}) can be extended to  standard intersection types.   It is useful here to add the notion of simple retraction.

\begin{definition} \label{defr}~
\begin{enumerate} \item Type $\sigma$ is a \emph{retract} of type $\tau$  (notation $\sigma \lhd \tau$) if there exist terms $L$ and $R$ such that:
\begin{enumerate}
\item \label{defr11} $\vdash L: \tau \to \sigma$;
\item \label{defr12} $\vdash R: \sigma \to \tau$;
\item \label{defr13} $L \circ R = {\bf I }$.
\end{enumerate}
\item
 Type $\sigma$ is a \emph{simple retract} of type $\tau$ (notation $\sigma \lhd_s \tau$) if $\sigma \lhd \tau$ and $R$ is a simple right inverse.
\end{enumerate}
 We say that $L,R$ witness the (simple) retraction.
\end{definition}

\medskip

We can now show the desired characterisation.
\begin{theorem}
  [Characterisation of Simple Retractions in Standard Intersection Types]\labelx{c2}
  $
  \sigma\triangleleft_s\tau
  $
  if and only if $\sigma\sim\bigwedge_{i\in I}\sigma_i$ and
 $
  \tau\sim\bigwedge_{i\in I}(\rho_1^{(i)}\to\dots\to\rho_m^{(i)}\to\sigma_i)
  $
  and the types $\tau\to\bigwedge_{i\in I}\rho_k^{(i)}$ are inhabited for $1\leq k\leq m$.
\end{theorem}
\begin{proof}~
  \begin{description}
    \item[(If)]
      Let $M_k$ be an inhabitant of $\tau\to\bigwedge_{i\in I}\rho_k^{(i)}$ for $1\leq k\leq m$.
      We can choose $L=\lambda z.z(M_1z)\dots(M_mz)$ and $R=\lambda tx_1\dots x_m.t$. It is easy to verify that $L$ and $R$ satisfy the conditions of Definition~\ref{defr}.
    \item[(Only if)] 
      Since $R$ is simple, we can assume $R=\lambda tx_1\dots x_m.t$. Lemma~\ref{ils}(\ref{ils3}) applied to $\vdash R: \sigma \to \tau$ gives $
      \tau\sim\bigwedge_{i\in I}(\rho_1^{(i)}\to\dots\to\rho_m^{(i)}\to\sigma_i)$ and $ t:\sigma\vdash t:\sigma_i$ for all $i\in I$. Then $\sigma\leq \sigma_i$ for all $i\in I$, which implies
      $\sigma\leq \bigwedge_{i\in I}\sigma_i$.

      By Theorem~\ref{right} we have $L=\lambda z.zM_1\dots M_m$.
      Then applying the Inversion Lemma to $\vdash L:\tau\to\sigma$ we get $\tau\leq \rho_1\to\dots\to\rho_m\to\sigma$  and $ z:\tau\vdash M_k:\rho_k$ for $1\leq k\leq m$. Then
      \[\bigwedge_{i\in I}(\rho_1^{(i)}\to\dots\to\rho_m^{(i)}\to\sigma_i)\leq \rho_1\to\dots\to\rho_m\to\sigma\] By Lemma~\ref{2.4(ii)} this implies $\bigwedge_{i\in J}\sigma_i\leq \sigma$ and $\rho_k\leq \bigwedge_{i\in J}\rho_k^{(i)}$ for some $J\subseteq I$ and for all $1\leq k\leq m$. From $\sigma\leq \bigwedge_{i\in I}\sigma_i$ and $\bigwedge_{i\in J}\sigma_i\leq \sigma$ with $J\subseteq I$ we get $J= I$ and $\sigma\sim \bigwedge_{i\in I}\sigma_i$. From $ z:\tau\vdash M_k:\rho_k$ and $\rho_k\leq \bigwedge_{i\in I}\rho_k^{(i)}$  we can derive $ z:\tau\vdash M_k:\bigwedge_{i\in I}\rho_k^{(i)}$ for $1\leq k\leq m$. Therefore the types $\tau\to\bigwedge_{i\in I}\rho_k^{(i)}$ are inhabited for $1\leq k\leq m$.
      \qedhere
  \end{description}
\end{proof}



In~\cite{B11} it is proved  that each intersection type is equivalent to an intersection of strict types. Owing to this property and the idempotence of the intersection type constructor, we can show that a simple retraction between intersection types implies a set of retractions between strict types.

\begin{corollary}
  If $
  \sigma\triangleleft_s\tau
  $ and $\tau\sim \bigwedge_{i\in I}\nu_i$, then there are strict types $\mu_i$ with $i\in I$ such that $\sigma\sim \bigwedge_{i\in I}\mu_i$,  and $\mu_i\triangleleft\nu_i$ for $i\in I$.
\end{corollary}
\begin{proof}By Theorem~\ref{c2} $\nu_i\sim\rho_1^{(i)}\to\dots\to\rho_m^{(i)}\to\mu_i$ for some $\mu_i$ and $\sigma\sim \bigwedge_{i\in I}\mu_i$.  Then $\mu_i\triangleleft\nu_i$ for $i\in I$ by Theorem~\ref{cec}.
\end{proof}

\begin{example} \label{excor}
  ~
  \begin{itemize}
    \item[i)] Consider 
	$\sigma=\varphi\wedge\psi$ ~and~ $\tau=\omega\to((\varphi\to\varphi\to\varphi)\to\varphi)\wedge((\psi\to\omega\to\psi)\to\psi)$. 
      We get $\sigma\triangleleft_s\tau$ by choosing $L= \lambda z. z \bot{\bf K}$, where ${\bf K}=\lambda xy. x$ and its simple right inverse $R=\lambda tx_1x_2.t$. The terms $L$, $R$ witness also the two retractions
      $\varphi \,\triangleleft\, \omega\to (\varphi\to\varphi\to\varphi)\to\varphi$ ~and~  $\psi\, \triangleleft \, \omega\to (\psi\to\omega\to\psi)\to\psi$ between strict types.
    \item[ii)]  The terms $L=\lambda z.z {\bf I}$ ~and~ $R=\lambda t x.t$ witness the simple retraction $\varphi\,\triangleleft_s((\psi\to\psi)\to \varphi)\wedge(\omega\to\varphi)$. The same terms show the two retractions $\varphi\,\triangleleft \,(\psi\to\psi)\to \varphi$ and $\varphi\,\triangleleft\,\omega\to\varphi$  between strict types.
  \end{itemize}
\end{example}

\section{Retractions in Models} \label{sem}

In this section we discuss retractions in models of $\lambda\bot$-calculus.
 Since standard intersection types can be seen as a conservative extension of strict intersection types, the results of this section hold for both systems. \\
 It is easy to adapt the Hindley-Longo definition of $\lambda$-calculus models~\cite{HL80}
 to the $\lambda\bot$-calculus. We use $\Lenv$ to range over  mappings from term variables to elements of the domain.

 \begin{definition} A model of the $\lambda\bot$-calculus is a structure $\Mod = <D, \cdot, \Gint{-}{\Mod}~, d_0 >$ which satisfies the following conditions:
 \begin{enumerate}
   \item\label{p8} $d_0 $ is a distinguished element of the domain $D$ such that $d_0 \cdot e = d_0$~ for all $e \in D$
   \item \label{p1} $\Gint{x}{\Mod}{\Lenv} = \Lenv(x)$
   \item \label{p2} $\Gint{M N}{\Mod}{\Lenv} = \Gint{M}{\Mod}{\Lenv} \cdot \Gint{N}{\Mod}{\Lenv}$
   \item \label{p3} $\Gint{\lambda x . M}{\Mod}{\Lenv} \cdot d = \Gint{ M}{\Mod}{\Lenv [d/x]}$
   \item \label{p4} $\Lenv = \Lenv'$ implies $\Gint{M}{\Mod}{\Lenv} = \Gint{M}{\Mod}{\Lenv'}$
   \item \label{p5} $\Gint{\lambda x . M}{\Mod}{\Lenv} = \Gint{\lambda y . M[y/x]}{\Mod}{\Lenv}$ if $y$ is not in $M$
   \item \label{p6} $\forall d \in D~~\Gint{ M}{\Mod}{\Lenv [d/x]} = \Gint{ N}{\Mod}{\Lenv [d/x]}$ implies
     $\Gint{\lambda x . M}{\Mod}{\Lenv} = \Gint{\lambda x . N}{\Mod}{\Lenv}$
   \item \label{p7} $\Gint{\bot}{\Mod}{\Lenv}= \Gint{\lambda x.\bot}{\Mod}{\Lenv} = d_0$.
 \end{enumerate}
 \end{definition}

 As usual we interpret types as subsets of the model domain.
 \begin{definition}
 The standard \emph{interpretation of types} in a model $\Mod$ with domain $D$ is defined by:
  \begin{enumerate}
   \item $\Gint{\omega}{\Mod}{\Tenv} = D$
  \item $\Gint{\varphi}{\Mod}{\Tenv} = \Tenv(\varphi)$
  \item $\Gint{\sigma \to \tau}{\Mod}{\Tenv} = \set{ d ~|~ \forall e \in \Gint{\sigma}{\Mod}{\Tenv}~~ d \cdot e \in \Gint{\tau}{\Mod}{\Tenv}$}
    \item $\Gint{\sigma \wedge \tau}{\Mod}{\Tenv} = \Gint{\sigma}{\Mod}{\Tenv} \cap \Gint{\tau}{\Mod}{\Tenv}$
  \end{enumerate}
  where $\Tenv$ ranges over mappings from type variables to subsets of the domain.
  \end{definition}

  From these definitions it easy to show the soundness of the type systems of Sections~\ref{it} and~\ref{git}.
  We say that the mappings $\Lenv,\Tenv$ for a model $\Mod$ respect an environment $\B$ if $x:\sigma\in\B$ implies
 $\Gint{x}{\Mod}{\Lenv} \in \Gint{\sigma}{\Mod}{\Tenv}$.

 \begin{theorem}[Soundness]\label{sa}If $\Gamma \vdash M : \sigma$, then $\Gint{M}{\Mod}{\Lenv} \in \Gint{\sigma}{\Mod}{\Tenv}$ for all $\Mod$ and all $\Lenv$, $\Tenv$ respecting $\Gamma$.
 \end{theorem}

  \smallskip

 The semantic retraction is naturally defined as follows.
  \begin{definition} \label{def1}
 Type $\sigma$ is a \emph{semantic retract} of type $\tau$   (notation $\sigma \blacktriangleleft \tau$) if there are two terms $L$ and $R$ such that for all models $\Mod$ of $\lambda\bot$-calculus and for all mappings $\Lenv,\Tenv$:
\begin{enumerate}
\item \label{d3} $\Gint{R}{\Mod}{\Lenv} \in \Gint{\sigma \to \tau}\Mod\Tenv$
\item \label{d4} $\Gint{L}{\Mod}{\Lenv} \in \Gint{\tau \to \sigma}\Mod\Tenv$
\item \label{d2} $\Gint{L}{\Mod}{\Lenv} \circ \Gint{R}{\Mod}{\Lenv}=\Gint{\bf I }{\Mod}{\Lenv}$.
\end{enumerate}
\end{definition}

\smallskip

 An interesting model of $\lambda$-calculus is the \emph{filter} model $\fmod$ defined in \cite{barecoppdeza83}. The domain of this model is the set of filters of types. We refer to that paper for the basic definitions and properties. The model $\fmod$ can be easily seen as a model of the $\lambda\bot$-calculus by taking the filter generated by type $\omega$ as $d_0$.

 In~\cite{barecoppdeza83} the filter model is proved to be complete for $\lambda$-calculus and standard intersection types. It is easy to adapt this result to  $\lambda\bot$-calculus. The completeness theorem of~\cite{barecoppdeza83} can then be reformulated as follows:

 \begin{theorem}[Completeness of the Filter Model]\label{fsc}If $\Gint{M}{\fmod}{\Lenv} \in \Gint{\sigma}{\fmod}{\Tenv}$ for all $\Lenv$, $\Tenv$ respecting $\Gamma$, then $\Gamma \vdash M : \sigma$.
 \end{theorem}

We can show that retraction and semantic retraction coincide using the soundness and the completeness of the filter model.

\begin{theorem}\label{csc} $\sigma \lhd \tau$ if and only if
$\sigma \blacktriangleleft \tau$.
\end{theorem}
\begin{proof}~
  \begin{description}
    \item[(If)]
    By Theorem~\ref{fsc} we immediately have that \begin{itemize} \item$\Gint{L}{\fmod}{\Lenv} \in \Gint{\tau\to\sigma}{\fmod}{\Tenv}$ for all $\Lenv$, $\Tenv$ implies \mbox{$\vdash L :\tau\to \sigma$} and \item $\Gint{R}{\fmod}{\Lenv} \in \Gint{\sigma\to\tau}{\fmod}{\Tenv}$ for all $\Lenv$, $\Tenv$ implies $\vdash R :\sigma\to \tau$. \end{itemize}Notice that $\Gint{L}{\Mod}{\Lenv} \circ \Gint{R}{\Mod}{\Lenv}=\Gint{{\bf B}LR }{\Mod}{\Lenv}$. Since $\varphi\to\varphi\in\Gint{\bf I }{\fmod}{\Lenv}$, the completeness of the filter model gives $\vdash {\bf B}LR: \varphi\to\varphi$. It is easy to prove that this implies ${\bf B}LR={\bf I }$, as remarked in~\cite{barecoppdeza83}.
    \item[(Only if)]
    The soundness of the type system (Theorem~\ref{sa}) implies that\begin{itemize} \item$\vdash L :\tau\to \sigma$ gives \mbox{$\Gint{L}{\Mod}{\Lenv} \in \Gint{\tau\to\sigma}{\fmod}{\Tenv}$} for all $\Mod$, $\Lenv$, $\Tenv$ and \item $\vdash R :\sigma\to \tau$ gives $\Gint{R}{\Mod}{\Lenv} \in \Gint{\sigma\to\tau}{\fmod}{\Tenv}$ for all $\Mod$, $\Lenv$, $\Tenv$. \end{itemize}If $L\circ R={\bf I }$, then by definition of model we get $\Gint{L}{\Mod}{\Lenv} \circ \Gint{R}{\Mod}{\Lenv}=\Gint{\bf I }{\Mod}{\Lenv}$.
      \qedhere
  \end{description}
\end{proof}

\section{Related Work}\label{rw}
The seminal paper~\cite{BL85} characterises for Curry types both isomorphism in the $\l\beta\eta$-calculus and retraction in the $\l\beta$-calculus. 

\medskip

Isomorphism in the $\l\beta\eta$-calculus is characterised for various type disciplines by means of equations between types~\cite{D05}. Product and unit types are considered in~\cite{S83,S93,BDL92}, universally quantified types in~\cite{BL85}, all the above type constructors in~\cite{D95}. Characterisation of isomorphism for intersection types instead requires a notion of type similarity~\cite{DDGT10,CDMZ14}. Analogous result holds for intersection and union types~\cite{CDMZ13,CDMZ15,CDMZ16}. \cite{FDB06} shows that isomorphism for product, arrow and sum types is not finitely axiomatizable. 

\medskip

Retraction witnessed by affine terms of $\l\beta\eta$-calculus for Curry types with only one atom is characterised in~\cite{LPS92}. As an auxiliary result in the study of the relation between iteration and recursion,~\cite{SU99} gives a necessary condition for retraction considering universally quantified types and  $\l\beta$-calculus. An algorithm to decide if a Curry type with a single atom is a retract of another one in the  $\l\beta\eta$-calculus is given in~\cite{P01}.  This algorithm builds the witnesses of the retraction, when they exist. The results of~\cite{LPS92} are extended  to Curry types with many atoms in~\cite{RU02}.  Moreover~\cite{RU02} gives necessary conditions for retraction witnessed by arbitrary terms of $\l\beta\eta$-calculus dealing with both Curry and universally quantified types. The problem of retraction solved in~\cite{LPS92} is shown to be NP-complete in~\cite{S08}. \cite{S13} gives a proof system which leads to an exponential decision procedure to characterise retraction for Curry types in the $\l\beta\eta$-calculus. 
\section{Conclusion}\label{c}
This paper deals with retraction for strict and standard intersection types in the $\lambda\bot$-calculus. Both the choices of the calculus  and of the types can be discussed.

\medskip

We considered the $\lambda\bot$-calculus following~\cite{MZ83}.
Our results are easily adapted to the $\lambda$-calculus taking an unsolvable term to play the role of $\bot$.

\medskip

By conservativity the given characterisation of retraction in strict intersection types holds in Curry types~\cite{currfeys58}. In this way we obtain the result of~\cite{BL85}. We give also a characterisation of retractions in standard intersection types when the right inverses are simple.  However the retraction of Example~\ref{ines} cannot be shown by a simple right inverse, since it should have both the types \\  \centerline{$\varphi\to\omega\to(\varphi\to\varphi)\to\varphi\text{ and }\varphi\to\omega\to(\varphi\to\varphi\to\varphi)\to\varphi\to\varphi$.}

 \medskip

The definition of a necessary and sufficient condition for the existence of a retraction between standard intersection types is not obvious when the right inverses are arbitrary. For example, no type can be a retract of the type\\  \centerline{$\omega\to((\omega\to\varphi\to\varphi)\to\varphi)\wedge((\psi\to\omega\to\psi)\to\psi)$} essentially since $(\omega\to\varphi\to\varphi)\wedge(\psi\to\omega\to\psi)$ is not inhabited. Notice that this type and the type $\tau$ of Example~\ref{excor}\,\emph{i)} differ only for one occurrence of $\omega$ in place of one occurrence of $\varphi$.

 \medskip

 We plan to investigate retractions in standard intersection types without conditions on right inverses, and in intersection and union types.
The problem for the $\l\beta\eta$-calculus is surely more difficult as shown by the papers~\cite{BL85,LPS92,P01,RU02,S08,S13} and it is left for future work.

\medskip

As suggested by one referee, an interesting future development is to see how the results presented here can be adapted to programming languages with richer sets of constructs. This would allow to apply automatic retraction inference in dealing with API of functional programs or proof-assistants so as to maximise code reuse.


\bigskip

\noindent {\bf Acknowledgments.}
We are grateful to the anonymous reviewers for their useful
suggestions, which led to substantial improvements. This work was done during a visit of Alejandro D\'{\i}az-Caro at the Computer Science Department of Torino University. This visit has been supported by the WWS2 Project, financed by ``Fondazione CRT" and Torino University.

\bibliographystyle{eptcs}
\bibliography{biblio}

\providecommand{\noopsort}[1]{}
\begin{thebibliography}{10}
\providecommand{\bibitemdeclare}[2]{}
\providecommand{\urlprefix}{Available at }
\providecommand{\url}[1]{\texttt{#1}}
\providecommand{\href}[2]{\texttt{#2}}
\providecommand{\urlalt}[2]{\href{#1}{#2}}
\providecommand{\doi}[1]{doi:\urlalt{http://dx.doi.org/#1}{#1}}
\providecommand{\bibinfo}[2]{#2}

\bibitemdeclare{article}{B11}
\bibitem{B11}
\bibinfo{author}{Steffen van Bakel} (\bibinfo{year}{2011}):
  \emph{\bibinfo{title}{{Strict Intersection Types for the Lambda Calculus}}}.
\newblock {\sl \bibinfo{journal}{{ACM Computing Surveys}}}
  \bibinfo{volume}{43}(\bibinfo{number}{3}), p.~\bibinfo{pages}{20},
  \doi{10.1145/1922649.1922657}.

\bibitemdeclare{book}{B84}
\bibitem{B84}
\bibinfo{author}{Henk Barendregt} (\bibinfo{year}{1984}):
  \emph{\bibinfo{title}{The Lambda Calculus: its Syntax and Semantics}},
  \bibinfo{edition}{revised} edition.
\newblock \bibinfo{publisher}{North-Holland}.

\bibitemdeclare{article}{barecoppdeza83}
\bibitem{barecoppdeza83}
\bibinfo{author}{Henk Barendregt}, \bibinfo{author}{Mario Coppo} \&
  \bibinfo{author}{Mariangiola Dezani-Ciancaglini} (\bibinfo{year}{1983}):
  \emph{\bibinfo{title}{{A Filter Lambda Model and the Completeness of Type
  Assignment}}}.
\newblock {\sl \bibinfo{journal}{The Journal of Symbolic Logic}}
  \bibinfo{volume}{48}(\bibinfo{number}{4}), pp. \bibinfo{pages}{931--940},
  \doi{10.2307/2273659}.

\bibitemdeclare{book}{BDS13}
\bibitem{BDS13}
\bibinfo{author}{Henk Barendregt}, \bibinfo{author}{Wil Dekkers} \&
  \bibinfo{author}{Richard Statman} (\bibinfo{year}{2013}):
  \emph{\bibinfo{title}{Lambda Calculus with Types}}.
\newblock \bibinfo{series}{Perspectives in logic},
  \bibinfo{publisher}{Cambridge University Press},
  \doi{10.1017/CBO9781139032636}.

\bibitemdeclare{inproceedings}{BD74}
\bibitem{BD74}
\bibinfo{author}{Corrado B{\"{o}}hm} \& \bibinfo{author}{Mariangiola
  Dezani{-}Ciancaglini} (\bibinfo{year}{1974}):
  \emph{\bibinfo{title}{{Combinatorial Problems, Combinator Equations and
  Normal Forms}}}.
\newblock In \bibinfo{editor}{Jacques Loeckx}, editor: {\sl
  \bibinfo{booktitle}{ICALP'74}}, {\sl
  \bibinfo{series}{LNCS}}~\bibinfo{volume}{14}, \bibinfo{publisher}{Springer},
  pp. \bibinfo{pages}{185--199}, \doi{10.1007/3-540-06841-4\_60}.

\bibitemdeclare{article}{BDL92}
\bibitem{BDL92}
\bibinfo{author}{Kim Bruce}, \bibinfo{author}{Roberto Di~Cosmo} \&
  \bibinfo{author}{Giuseppe Longo} (\bibinfo{year}{1992}):
  \emph{\bibinfo{title}{Provable isomorphisms of types}}.
\newblock {\sl \bibinfo{journal}{Mathematical Structures in Computer Science}}
  \bibinfo{volume}{2}(\bibinfo{number}{2}), pp. \bibinfo{pages}{231--247},
  \doi{10.1017/S0960129500001444}.

\bibitemdeclare{inproceedings}{BL85}
\bibitem{BL85}
\bibinfo{author}{Kim Bruce} \& \bibinfo{author}{Giuseppe Longo}
  (\bibinfo{year}{1985}): \emph{\bibinfo{title}{{Provable Isomorphisms and
  Domain Equations in Models of Typed Languages}}}.
\newblock In \bibinfo{editor}{Robert Sedgewick}, editor: {\sl
  \bibinfo{booktitle}{STOC'85}}, \bibinfo{publisher}{ACM Press}, pp.
  \bibinfo{pages}{263 -- 272}, \doi{10.1145/22145.22175}.

\bibitemdeclare{inproceedings}{CDMZ13}
\bibitem{CDMZ13}
\bibinfo{author}{Mario Coppo}, \bibinfo{author}{Mariangiola
  Dezani-Ciancaglini}, \bibinfo{author}{Ines Margaria} \&
  \bibinfo{author}{Maddalena Zacchi} (\bibinfo{year}{2013}):
  \emph{\bibinfo{title}{{Towards Isomorphism of Intersection and Union
  Types}}}.
\newblock In \bibinfo{editor}{Stephane Graham-Lengrand} \&
  \bibinfo{editor}{Luca Paolini}, editors: {\sl \bibinfo{booktitle}{ITRS'12}},
  {\sl \bibinfo{series}{EPTCS}} \bibinfo{volume}{121}, pp. \bibinfo{pages}{58
  -- 80}, \doi{10.4204/EPTCS.121.5}.

\bibitemdeclare{inproceedings}{CDMZ14}
\bibitem{CDMZ14}
\bibinfo{author}{Mario Coppo}, \bibinfo{author}{Mariangiola
  Dezani-Ciancaglini}, \bibinfo{author}{Ines Margaria} \&
  \bibinfo{author}{Maddalena Zacchi} (\bibinfo{year}{2014}):
  \emph{\bibinfo{title}{Isomorphism of ``Functional" Intersection Types}}.
\newblock In \bibinfo{editor}{Ralph Matthes} \& \bibinfo{editor}{Aleksy
  Schubert}, editors: {\sl \bibinfo{booktitle}{Types'13}},
  \bibinfo{volume}{26}, \bibinfo{publisher}{LIPIcs}, pp.
  \bibinfo{pages}{129--149}, \doi{10.4230/LIPIcs.TYPES.2013.129}.

\bibitemdeclare{article}{CDMZ16}
\bibitem{CDMZ16}
\bibinfo{author}{Mario Coppo}, \bibinfo{author}{Mariangiola
  Dezani-Ciancaglini}, \bibinfo{author}{Ines Margaria} \&
  \bibinfo{author}{Maddalena Zacchi} (\bibinfo{year}{2015}):
  \emph{\bibinfo{title}{{Isomorphism of Intersection and Union Types}}}.
\newblock {\sl \bibinfo{journal}{Mathematical Structures in Computer Science}}
  \doi{10.1017/S0960129515000304}.
\newblock \bibinfo{note}{Published online: 07 August 2015}.

\bibitemdeclare{inproceedings}{CDMZ15}
\bibitem{CDMZ15}
\bibinfo{author}{Mario Coppo}, \bibinfo{author}{Mariangiola
  Dezani-Ciancaglini}, \bibinfo{author}{Ines Margaria} \&
  \bibinfo{author}{Maddalena Zacchi} (\bibinfo{year}{2015}):
  \emph{\bibinfo{title}{{On Isomorphism of ``Functional'' Intersection and
  Union Types}}}.
\newblock In \bibinfo{editor}{Jacob Rehof}, editor: {\sl
  \bibinfo{booktitle}{ITRS'14}}, {\sl \bibinfo{series}{EPTCS}}
  \bibinfo{volume}{177}, pp. \bibinfo{pages}{53--64}, \doi{10.4204/EPTCS.177}.

\bibitemdeclare{book}{currfeys58}
\bibitem{currfeys58}
\bibinfo{author}{Haskell~B. Curry} \& \bibinfo{author}{Robert Feys}
  (\bibinfo{year}{1958}): \emph{\bibinfo{title}{Combinatory Logic}}.
\newblock {\sl \bibinfo{series}{Studies in Logic and the Foundations of
  Mathematics}}~\bibinfo{volume}{I}, \bibinfo{publisher}{North-Holland}.

\bibitemdeclare{inproceedings}{LPS92}
\bibitem{LPS92}
\bibinfo{author}{Ugo de'Liguoro}, \bibinfo{author}{Adolfo Piperno} \&
  \bibinfo{author}{Rick Statman} (\bibinfo{year}{1992}):
  \emph{\bibinfo{title}{{Retracts in Simply Typed $\l\beta\eta$-calculus}}}.
\newblock In \bibinfo{editor}{Andre Scedrov}, editor: {\sl
  \bibinfo{booktitle}{LICS'92}}, \bibinfo{publisher}{IEEE Computer Society
  Press}, pp. \bibinfo{pages}{461--469}, \doi{10.1109/LICS.1992.185557}.

\bibitemdeclare{article}{DDGT10}
\bibitem{DDGT10}
\bibinfo{author}{Mariangiola Dezani-Ciancaglini}, \bibinfo{author}{Roberto~Di
  Cosmo}, \bibinfo{author}{Elio Giovannetti} \& \bibinfo{author}{Makoto
  Tatsuta} (\bibinfo{year}{2010}): \emph{\bibinfo{title}{On Isomorphisms of
  Intersection Types}}.
\newblock {\sl \bibinfo{journal}{ACM Transactions on Computational Logic}}
  \bibinfo{volume}{11}(\bibinfo{number}{4}), pp. \bibinfo{pages}{1--22},
  \doi{10.1145/1805950.1805955}.

\bibitemdeclare{article}{D95}
\bibitem{D95}
\bibinfo{author}{Roberto Di~Cosmo} (\bibinfo{year}{1995}):
  \emph{\bibinfo{title}{{Second Order Isomorphic Types. A Proof Theoretic Study
  on Second Order $\lambda$-calculus with Surjective Pairing and Terminal
  Object}}}.
\newblock {\sl \bibinfo{journal}{Information and Computation}}
  \bibinfo{volume}{119}(\bibinfo{number}{2}), pp. \bibinfo{pages}{176--201},
  \doi{10.1006/inco.1995.1085}.

\bibitemdeclare{article}{D05}
\bibitem{D05}
\bibinfo{author}{Roberto Di~Cosmo} (\bibinfo{year}{2005}):
  \emph{\bibinfo{title}{{A Short Survey of Isomorphisms of Types}}}.
\newblock {\sl \bibinfo{journal}{Mathematical Structures in Computer Science}}
  \bibinfo{volume}{15}, pp. \bibinfo{pages}{825--838},
  \doi{10.1017/S0960129505004871}.

\bibitemdeclare{article}{FDB06}
\bibitem{FDB06}
\bibinfo{author}{Marcelo Fiore}, \bibinfo{author}{Roberto~Di Cosmo} \&
  \bibinfo{author}{Vincent Balat} (\bibinfo{year}{2006}):
  \emph{\bibinfo{title}{{Remarks on Isomorphisms in Typed Lambda Calculi with
  Empty and Sum Types}}}.
\newblock {\sl \bibinfo{journal}{Annals of Pure and Applied Logic}}
  \bibinfo{volume}{141}(\bibinfo{number}{1--2}), pp. \bibinfo{pages}{35--50},
  \doi{10.1016/j.apal.2005.09.001}.

\bibitemdeclare{article}{HL80}
\bibitem{HL80}
\bibinfo{author}{Roger Hindley} \& \bibinfo{author}{Giuseppe Longo}
  (\bibinfo{year}{1980}): \emph{\bibinfo{title}{Lambda-Calculus Models and
  Extensionality}}.
\newblock {\sl \bibinfo{journal}{Mathematical Logic Quarterly}}
  \bibinfo{volume}{26}(\bibinfo{number}{19-21}), pp. \bibinfo{pages}{289--310},
  \doi{10.1002/malq.19800261902}.

\bibitemdeclare{article}{MZ83}
\bibitem{MZ83}
\bibinfo{author}{Ines Margaria} \& \bibinfo{author}{Maddalena Zacchi}
  (\bibinfo{year}{1983}): \emph{\bibinfo{title}{{Right and Left Invertibility
  in Lambda-Beta-Calculus}}}.
\newblock {\sl \bibinfo{journal}{R.A.I.R.O. Theoretical Informatics}}
  \bibinfo{volume}{17}(\bibinfo{number}{1}), pp. \bibinfo{pages}{71--88}.

\bibitemdeclare{inproceedings}{P01}
\bibitem{P01}
\bibinfo{author}{Vincent Padovani} (\bibinfo{year}{2001}):
  \emph{\bibinfo{title}{{Retracts in Simple Types}}}.
\newblock In \bibinfo{editor}{Samson Abramsky}, editor: {\sl
  \bibinfo{booktitle}{TLCA'01}}, {\sl \bibinfo{series}{LNCS}}
  \bibinfo{volume}{2044}, \bibinfo{publisher}{Springer}, pp.
  \bibinfo{pages}{376--384}, \doi{10.1007/3-540-45413-6\_29}.

\bibitemdeclare{article}{RU02}
\bibitem{RU02}
\bibinfo{author}{Laurent Regnier} \& \bibinfo{author}{Pawel Urzyczyn}
  (\bibinfo{year}{2002}): \emph{\bibinfo{title}{{Retractions of Types with Many
  Atoms}}}.
\newblock {\sl \bibinfo{journal}{CoRR}} \bibinfo{volume}{cs.LO/0212005}.

\bibitemdeclare{article}{S08}
\bibitem{S08}
\bibinfo{author}{Aleksy Schubert} (\bibinfo{year}{2008}):
  \emph{\bibinfo{title}{{On the Building of Affine Retractions}}}.
\newblock {\sl \bibinfo{journal}{Mathematical Structures in Computer Science}}
  \bibinfo{volume}{18}(\bibinfo{number}{4}), pp. \bibinfo{pages}{753--793},
  \doi{10.1017/S096012950800683X}.

\bibitemdeclare{article}{S83}
\bibitem{S83}
\bibinfo{author}{Sergei Soloviev} (\bibinfo{year}{1983}):
  \emph{\bibinfo{title}{{The Category of Finite Sets and Cartesian Closed
  Categories}}}.
\newblock {\sl \bibinfo{journal}{Journal of Soviet Mathematics}}
  \bibinfo{volume}{22}(\bibinfo{number}{3}), pp. \bibinfo{pages}{1387--1400},
  \doi{10.1007/BF01084396}.
\newblock \bibinfo{note}{English translation of the original paper in Russian
  published in Zapiski Nauchnykh Seminarov LOMI, v.105, 1981}.

\bibitemdeclare{inproceedings}{S93}
\bibitem{S93}
\bibinfo{author}{Sergei Soloviev} (\bibinfo{year}{1993}):
  \emph{\bibinfo{title}{{A Complete Axiom System for Isomorphism of Types in
  Closed Categories}}}.
\newblock In \bibinfo{editor}{Andrei Voronkov}, editor: {\sl
  \bibinfo{booktitle}{LPAR'93}}, {\sl \bibinfo{series}{LNCS}}
  \bibinfo{volume}{698}, \bibinfo{publisher}{Springer}, pp.
  \bibinfo{pages}{360--371}, \doi{10.1007/3-540-56944-8\_71}.

\bibitemdeclare{inproceedings}{SU99}
\bibitem{SU99}
\bibinfo{author}{Zdzislaw Splawski} \& \bibinfo{author}{Pawel Urzyczyn}
  (\bibinfo{year}{1999}): \emph{\bibinfo{title}{{Type Fixpoints: Iteration vs.
  Recursion}}}.
\newblock In \bibinfo{editor}{Didier R{\'{e}}mi} \& \bibinfo{editor}{Peter
  Lee}, editors: {\sl \bibinfo{booktitle}{ICFP '99}}, \bibinfo{publisher}{ACM
  Press}, pp. \bibinfo{pages}{102--113}, \doi{10.1145/317636.317789}.

\bibitemdeclare{inproceedings}{S13}
\bibitem{S13}
\bibinfo{author}{Colin Stirling} (\bibinfo{year}{2013}):
  \emph{\bibinfo{title}{Proof Systems for Retracts in Simply Typed Lambda
  Calculus}}.
\newblock In \bibinfo{editor}{Fedor~V. Fomin}, \bibinfo{editor}{Rusins
  Freivalds}, \bibinfo{editor}{Marta~Z. Kwiatkowska} \& \bibinfo{editor}{David
  Peleg}, editors: {\sl \bibinfo{booktitle}{ICALP'13}}, {\sl
  \bibinfo{series}{LNCS}} \bibinfo{volume}{7966},
  \bibinfo{publisher}{Springer}, pp. \bibinfo{pages}{398--409},
  \doi{10.1007/978-3-642-39212-2\_36}.

\bibitemdeclare{article}{U99}
\bibitem{U99}
\bibinfo{author}{Pawel Urzyczyn} (\bibinfo{year}{1999}):
  \emph{\bibinfo{title}{{The Emptiness Problem for Intersection Types}}}.
\newblock {\sl \bibinfo{journal}{The Journal of Symbolic Logic}}
  \bibinfo{volume}{64}(\bibinfo{number}{3}), pp. \bibinfo{pages}{1195--1215},
  \doi{10.2307/2586625}.

\bibitemdeclare{inproceedings}{U09}
\bibitem{U09}
\bibinfo{author}{Pawel Urzyczyn} (\bibinfo{year}{2009}):
  \emph{\bibinfo{title}{{Inhabitation of Low-Rank Intersection Types}}}.
\newblock In \bibinfo{editor}{Pierre{-}Louis Curien}, editor: {\sl
  \bibinfo{booktitle}{TLCA}}, {\sl \bibinfo{series}{LNCS}}
  \bibinfo{volume}{5608}, \bibinfo{publisher}{Springer}, pp.
  \bibinfo{pages}{356--370}, \doi{10.1007/978-3-642-02273-9\_26}.

\end{thebibliography}

\end{document}